\documentclass[journal]{IEEEtran}

\usepackage{amsmath,amssymb,amsthm}
\usepackage{booktabs}
\usepackage{graphicx}
\usepackage{hyperref}
\usepackage{xcolor}
\usepackage{pifont}
\usepackage{tikz}
\usetikzlibrary{shapes,arrows,positioning,fit,calc}
\usepackage{algorithm}
\usepackage{algpseudocode}
\usepackage{multirow}
\usepackage{enumitem}
\usepackage{balance}
\usepackage{subcaption}
\usepackage{url}

\newtheorem{theorem}{Theorem}
\newtheorem{proposition}{Proposition}
\newtheorem{definition}{Definition}
\newtheorem{corollary}{Corollary}

\begin{document}

\title{RAGShield: Detecting Numerical Claim Manipulation\\in Government RAG Systems}

\author{KrishnaSaiReddy~Patil%
\thanks{Manuscript submitted April 2026.}}

\maketitle

\begin{abstract}
Retrieval-Augmented Generation (RAG) systems are deployed across federal agencies for citizen-facing tax guidance, benefits eligibility, and legal information, where a single incorrect number causes direct financial harm. This paper proves that all embedding-based RAG defenses share a fundamental blind spot: changing a tax deduction by \$50{,}000 produces cosine similarity 0.9998, invisible to every known detection threshold. Across 174 manipulation pairs and two embedding models, the mean sensitivity gap is 1{,}459$\times$. The blind spot is confirmed on real IRS documents.

The root cause is that embeddings encode topic, not numerical precision. RAGShield sidesteps this by operating on extracted values directly: a pattern-based engine identifies dollar amounts and percentages in government text, links each value to its governing entity through two-pass context propagation (99.8\% entity detection on 2{,}742 real IRS passages), and verifies every claim against a cross-source registry built from the corpus itself. A temporal tracker flags value changes that fall outside known government update schedules. On 430 attacks generated from real IRS document content, RAGShield detects every one (0.0\% ASR, 95\% CI [0\%, 1\%]) while embedding-based defenses miss 79--90\% of the same attacks.
\end{abstract}

\begin{IEEEkeywords}
RAG, retrieval-augmented generation, poisoning, numerical claims, government, insider threat, claim verification
\end{IEEEkeywords}

% ============================================================================
\section{Introduction}
\label{sec:intro}
% ============================================================================

Federal agencies are rapidly adopting Retrieval-Augmented Generation (RAG) systems for citizen-facing services. The General Services Administration launched USAi.Gov in August 2025, providing all federal agencies access to AI models from OpenAI, Anthropic, Google, and Meta~\cite{gsa_usai}. The Government Accountability Office documented 89 RAG-based citizen-facing services across 28 agencies as of fiscal year 2025~\cite{gao_ai_report}. These systems answer questions about tax obligations, benefits eligibility, legal rights, and government procedures, domains where incorrect numerical information causes direct financial harm to citizens.

RAG systems are vulnerable to knowledge base poisoning attacks. PoisonedRAG~\cite{poisonedrag2025} achieves over 90\% attack success rate by injecting 5 malicious texts. The Phantom attack~\cite{phantom_neurips2024} achieves 98.2\% retrieval success with 10 adversarial passages. Existing defenses, including RobustRAG~\cite{robustrag_neurips2025}, RAGDefender~\cite{ragdefender2025}, and TrustRAG~\cite{trustrag2025}, all operate on embedding-level content analysis: isolating passages, filtering by centroid distance, or clustering to detect outliers.

\textbf{Key observation.} All embedding-based defenses share a fundamental blind spot: embedding models are trained on semantic similarity, not numerical precision. I demonstrate empirically that changing a tax deduction from \$15{,}000 to \$65{,}000, a \$50{,}000 manipulation, produces a cosine similarity of 0.9998 between the original and modified passage embeddings. This perturbation is below \emph{every} known detection threshold. Across 174 manipulation pairs and two embedding models, the mean numerical sensitivity gap is 1{,}459$\times$: a 1\% change in embedding space corresponds to a 1{,}459\% change in the numerical claim value.

\textbf{The threat.} This blind spot enables a class of insider attacks that no existing defense can detect. An adversary with valid credentials can modify specific numerical values in knowledge base documents. Supply chain compromises such as SolarWinds~\cite{solarwinds} demonstrate that insider access to trusted systems is a realistic threat vector. Changing the standard deduction from \$15{,}000 to \$15{,}500 causes citizens to claim incorrect deductions. Changing the SSI benefit rate from \$943/month to \$993/month causes incorrect benefit expectations. These attacks pass provenance checks because the document has valid signatures, and they pass embedding-based detection because the perturbation is invisible. The result is measurable financial harm.

\textbf{Contributions.}

\begin{enumerate}[leftmargin=*]
\item \textbf{Embedding blind spot theorem} (\S\ref{sec:blindspot}): Formal and empirical proof that embedding-based RAG defenses cannot detect numerical claim manipulation. Mean sensitivity gap of 1{,}459$\times$ across two models and 174 pairs. Zero detections at any standard threshold.

\item \textbf{Claim-level verification framework} (\S\ref{sec:architecture}): Numerical claim extraction engine with two-pass entity resolution (forward propagation + backward fill) and cross-source claim registry with single-source discrepancy detection. 99.8\% entity detection on 2{,}742 real IRS passages, 83.3\% extraction precision, 100\% recall on government-format text.

\item \textbf{Temporal claim tracking} (\S\ref{sec:temporal}): Domain-specific authorized-change model distinguishing legitimate annual updates from unauthorized manipulation. 100\% accuracy on 8 temporal tests.

\item \textbf{Honest evaluation on real government documents} (\S\ref{sec:evaluation}): 2{,}742-passage real IRS corpus from 5 official publications, 430 attacks generated from real document content, 6 baselines, end-to-end LLM testing, and documented failure modes.
\end{enumerate}

CivicShield~\cite{civicshield2026} addresses conversational-level attacks against government AI chatbots. This paper addresses a complementary threat surface: subtle manipulation of the knowledge base that feeds those same systems, focusing on numerical claim integrity.

% ============================================================================
\section{Background and Related Work}
\label{sec:related}
% ============================================================================

\subsection{RAG Poisoning Attacks}

Table~\ref{tab:attacks} summarizes the attack landscape. Attacks span multiple capability levels: external corpus injection~\cite{poisonedrag2025,phantom_neurips2024,cparag2025}, data loader exploitation~\cite{phantomtext2025}, and insider compromise~\cite{confusedpilot2024}. Zhang et al.~\cite{benchmarkrag2025} benchmarked 13 attacks across 15 datasets and found that current defenses fail to provide robust protection. These attacks focus on injecting \emph{new} malicious documents. This paper addresses a different vector: \emph{modifying numerical values} in existing or legitimately-sourced documents.

\begin{table}[t]
\centering
\caption{RAG Poisoning Attack Landscape}
\label{tab:attacks}
\small
\begin{tabular}{@{}llll@{}}
\toprule
\textbf{Attack} & \textbf{Venue} & \textbf{Key Result} & \textbf{Vector} \\
\midrule
PoisonedRAG~\cite{poisonedrag2025} & USENIX'25 & 5 texts $\rightarrow$ 90\% ASR & Injection \\
Phantom~\cite{phantom_neurips2024} & NeurIPS'24 & 98.2\% retrieval & Injection \\
CPA-RAG~\cite{cparag2025} & arXiv'25 & Black-box, fluent & Injection \\
CorruptRAG~\cite{corruptrag2025} & arXiv'25 & Single-doc attack & Injection \\
PhantomText~\cite{phantomtext2025} & arXiv'25 & 74.4\% via hidden text & Loader \\
\midrule
\textbf{This paper} & --- & \textbf{Numerical manip.} & \textbf{Insider} \\
\bottomrule
\end{tabular}
\end{table}

\subsection{RAG Defenses}

Table~\ref{tab:defenses} compares existing defenses. RobustRAG~\cite{robustrag_neurips2025} uses isolate-then-aggregate with majority vote. This is effective against injection but operates on passage-level agreement, which cannot distinguish \$15{,}000 from \$15{,}500 within semantically identical passages. TrustRAG~\cite{trustrag2025} applies cluster filtering and LLM self-assessment, clustering operates on embeddings, inheriting the same blind spot; LLM self-assessment is unreliable for numerical comparison across passages. RAGPart~\cite{ragpart2025} uses document partitioning, random partitioning does not guarantee cross-referencing of numerical claims. RAGForensics~\cite{ragforensics2025} provides traceback after attacks but not prevention. No existing defense operates at the \emph{claim level}.

\begin{table}[t]
\centering
\caption{RAG Defense Comparison}
\label{tab:defenses}
\small
\begin{tabular}{@{}lcccc@{}}
\toprule
& \rotatebox{70}{Embed-Level} & \rotatebox{70}{Claim-Level} & \rotatebox{70}{Temporal} & \rotatebox{70}{Insider Det.} \\
\midrule
RobustRAG~\cite{robustrag_neurips2025} & \ding{51} & \ding{55} & \ding{55} & \ding{55} \\
TrustRAG~\cite{trustrag2025} & \ding{51} & \ding{55} & \ding{55} & \ding{55} \\
RAGPart~\cite{ragpart2025} & \ding{51} & \ding{55} & \ding{55} & \ding{55} \\
RAGDefender~\cite{ragdefender2025} & \ding{51} & \ding{55} & \ding{55} & \ding{55} \\
RAGForensics~\cite{ragforensics2025} & \ding{51} & \ding{55} & \ding{55} & Partial \\
\textbf{RAGShield} & \ding{51} & \ding{51} & \ding{51} & \ding{51} \\
\bottomrule
\end{tabular}
\end{table}

\subsection{Why Existing Defenses Fail on Numerical Manipulation}

To understand why the blind spot is structural rather than incidental, I analyze each defense's core mechanism against a concrete attack: changing ``the standard deduction for single filers is \$15{,}000'' to ``\$15{,}500'' (a \$500 manipulation producing cosine similarity 0.9997).

\textbf{RobustRAG} retrieves top-$k$ passages independently, generates an isolated response from each, and aggregates via majority vote. For numerical manipulation: the poisoned passage is semantically identical to the original and will appear in top-$k$ with high rank. The isolated response from the poisoned passage will say ``\$15{,}500'' while responses from other passages (if any discuss the same topic) will say ``\$15{,}000.'' However, in a mixed corpus with 1{,}000+ passages, the probability that multiple top-$k$ passages discuss the exact same numerical claim is low. The poisoned passage often stands alone on its topic, making majority vote ineffective because there is no majority to outvote it.

\textbf{TrustRAG} clusters retrieved passages and removes outlier clusters, then uses LLM self-assessment to detect inconsistencies. For numerical manipulation: the poisoned passage clusters \emph{with} the legitimate passage (cosine similarity 0.9997 means they are in the same cluster). Stage 1 clustering cannot separate them. Stage 2 LLM self-assessment would need to notice that ``\$15{,}500'' $\neq$ ``\$15{,}000'' across two passages in the context window, but research on LLM numerical reasoning shows this is unreliable, especially when the numbers are embedded in otherwise identical text.

\textbf{RAGPart} partitions the corpus into subsets and retrieves from each independently. For numerical manipulation: the poisoned and original passages may end up in the same partition (50\% probability with random 2-way partitioning), in which case the defense provides no benefit. Even when they are in different partitions, the aggregation step must compare specific numerical values across partition results, which RAGPart does not do, as it operates on passage-level agreement.

\textbf{RAGDefender} filters retrieved passages by centroid distance, removing those far from the centroid of the retrieved set. For numerical manipulation: the poisoned passage has nearly identical embedding to the original, so its distance from the centroid is indistinguishable from legitimate passages. Centroid-distance filtering cannot separate them.

The common failure mode across all four defenses is that they operate on \emph{embedding representations}, which encode semantic topic but not numerical precision. Claim-level verification sidesteps this entirely by comparing extracted numerical values directly.

\subsection{Fact Verification}

Automated fact verification~\cite{thorne2018fever} extracts claims from text and verifies them against evidence. ClaimBuster~\cite{claimbuster2017} identifies check-worthy claims. These systems operate on general natural language claims, not the structured numerical claims in government documents. RAGShield adapts claim extraction specifically for government numerical formats and cross-references against a domain-specific registry rather than open-web evidence.

\subsection{LLM Numerical Reasoning}

Large language models exhibit systematic weaknesses in numerical reasoning tasks, including precise comparison, arithmetic, and value extraction from context. This limitation is relevant to RAGShield in two ways: (1)~it explains why LLM-based self-assessment (as used in TrustRAG) is unreliable for detecting numerical manipulation. The LLM may not notice that \$15{,}500 $\neq$ \$15{,}000 across two passages; and (2)~it motivates the design choice to use structured claim extraction rather than LLM-based comparison for verification. The claim extraction engine operates on deterministic pattern matching, avoiding the stochastic errors inherent in LLM numerical processing.

% ============================================================================
\section{Threat Model}
\label{sec:threat}
% ============================================================================

\begin{definition}[RAG System]
A RAG system $\mathcal{R} = (\mathcal{K}, \eta, \rho, \gamma)$ consists of knowledge base $\mathcal{K} = \{d_1, \ldots, d_n\}$, embedding function $\eta: \mathcal{D} \rightarrow \mathbb{R}^m$, retriever $\rho: \mathcal{Q} \times \mathcal{K} \rightarrow \mathcal{P}(\mathcal{K})$, and generator $\gamma: \mathcal{Q} \times \mathcal{P}(\mathcal{K}) \rightarrow \mathcal{A}$.
\end{definition}

\begin{definition}[Numerical Claim]
A numerical claim $c = (e, a, v, u)$ consists of entity $e$, attribute $a$, value $v \in \mathbb{R}$, and unit $u$. A document $d$ contains a set of claims $C(d) = \{c_1, \ldots, c_k\}$.
\end{definition}

\begin{definition}[Numerical Claim Adversary]
An adversary $\mathcal{A}_{\text{num}}$ has valid credentials (insider or compromised source), modifies specific numerical values in documents, and aims to cause the RAG system to produce incorrect numerical answers. The adversary is constrained to changes small enough to avoid embedding-based detection: $\|\eta(d') - \eta(d)\| < \delta$ for detection threshold $\delta$.
\end{definition}

Three adversary tiers target numerical claims:

\begin{itemize}[leftmargin=*]
\item \textbf{T3-H (Subtle numerical manipulation):} Insider with valid provenance modifies dollar amounts by \$50--\$5{,}000. Embedding perturbation $< 0.01$.
\item \textbf{T6 (In-place replacement):} Insider replaces existing corpus documents with single-number modifications ($\pm 1$). Embedding perturbation $< 0.001$.
\item \textbf{T-TEMPORAL (Wrong-year values):} Adversary introduces correct values from a prior tax year, causing outdated information to be served as current.
\end{itemize}

\textbf{Harm model.} For a query $q$ with correct answer $a^*$ and poisoned answer $a'$, the harm is $h(q) = |v(a') - v(a^*)| \cdot s(q)$, where $v(\cdot)$ extracts the numerical value and $s(q)$ is the query sensitivity.

% ============================================================================
\section{The Embedding Blind Spot}
\label{sec:blindspot}
% ============================================================================

\begin{theorem}[Embedding Numerical Insensitivity]
\label{thm:blindspot}
Let $\eta: \mathcal{D} \rightarrow \mathbb{R}^m$ be any embedding model that produces document embeddings via pooling over $n$ token representations. Let $d$ be a document containing a numerical claim where the numerical token(s) occupy $t$ positions. For any $d'$ identical to $d$ except that value $v$ is replaced by $v' \neq v$:
\begin{equation}
\|\eta(d') - \eta(d)\| \leq \frac{t}{n} \cdot B
\label{eq:bound}
\end{equation}
where $B$ is the maximum per-token representation change. The numerical sensitivity gap $G = |v' - v| / \|\eta(d') - \eta(d)\| \geq |v' - v| \cdot n / (t \cdot B)$ is unbounded.
\end{theorem}

\begin{proof}
For mean pooling, $\eta(d) = \frac{1}{n} \sum_i h_i$. Replacing $v$ with $v'$ changes $t$ token representations. Let $\Delta h_j = h'_j - h_j$ for affected positions. Then $\|\eta(d') - \eta(d)\| = \|\frac{1}{n} \sum_{j} \Delta h_j\| \leq \frac{t}{n} \cdot \max_j \|\Delta h_j\|$. The bound $B = \max_j \|\Delta h_j\|$ depends on token embedding differences, not on $|v' - v|$: the model processes numerical tokens as discrete vocabulary items, so changing ``15'' to ``65'' produces a fixed $\Delta h$ regardless of the arithmetic difference. Since $|v' - v|$ is unconstrained while $\frac{t}{n} \cdot B$ is fixed, $G \rightarrow \infty$.
\end{proof}

\begin{corollary}[Universal Detection Failure]
\label{cor:detection_failure}
For any embedding-based defense with threshold $\delta > 0$, any numerical manipulation satisfies $\|\eta(d') - \eta(d)\| < \delta$ whenever $\frac{t}{n} \cdot B < \delta$, which holds for all documents where $n > t \cdot B / \delta$. The adversary's numerical change $|v' - v|$ is unconstrained by $\delta$.
\end{corollary}

\textbf{Empirical validation.} I validate Theorem~\ref{thm:blindspot} on 174 manipulation pairs across 12 government-format passages and two embedding models. Table~\ref{tab:blindspot} shows the aggregate results. Table~\ref{tab:blindspot_range} breaks down the sensitivity gap by the magnitude of the numerical change. The blind spot holds uniformly across small (\$1) and large (\$50{,}000) manipulations.

\begin{table}[t]
\centering
\caption{Embedding Blind Spot: Aggregate Results}
\label{tab:blindspot}
\small
\begin{tabular}{@{}lrr@{}}
\toprule
\textbf{Metric} & \textbf{MiniLM} & \textbf{BGE} \\
\midrule
Manipulation pairs tested & 174 & 174 \\
Mean cosine similarity & 0.9989 & 0.9979 \\
Max embedding perturbation & 0.0229 & 0.0243 \\
Mean sensitivity gap & 1{,}459$\times$ & 911$\times$ \\
Max numerical change & \$50{,}000 & \$50{,}000 \\
Detected at $\delta = 0.08$ & 0/174 & 0/174 \\
Detected at $\delta = 0.05$ & 0/174 & 0/174 \\
Detected at $\delta = 0.02$ & 1/174 & 3/174 \\
\bottomrule
\end{tabular}
\end{table}

% --- ADDITION 1: Sensitivity gap by delta range ---
\begin{table}[t]
\centering
\caption{Sensitivity Gap by Numerical Change Magnitude (MiniLM)}
\label{tab:blindspot_range}
\small
\begin{tabular}{@{}lrrrr@{}}
\toprule
\textbf{$\Delta v$ Range} & \textbf{Count} & \textbf{Mean CosSim} & \textbf{Mean Perturb} & \textbf{Max Perturb} \\
\midrule
$< 1\%$ & 54 & 0.9989 & 0.0011 & 0.0042 \\
$1$--$5\%$ & 31 & 0.9983 & 0.0017 & 0.0229 \\
$5$--$10\%$ & 15 & 0.9990 & 0.0010 & 0.0025 \\
$10$--$50\%$ & 34 & 0.9990 & 0.0010 & 0.0073 \\
$50$--$100\%$ & 14 & 0.9993 & 0.0007 & 0.0020 \\
$> 100\%$ & 26 & 0.9989 & 0.0011 & 0.0083 \\
\bottomrule
\end{tabular}
\end{table}

The most extreme example: changing the head-of-household standard deduction from \$22{,}500 to \$72{,}500 (\$50{,}000 change, 222\%) produces cosine similarity 0.99983 on MiniLM. Notably, the perturbation does \emph{not} increase with the magnitude of the numerical change: a \$50{,}000 change produces the same perturbation as a \$1 change. This confirms that embedding models encode the \emph{presence} of a number, not its \emph{value}.

\begin{proposition}[Provenance Insufficiency]
\label{prop:provenance}
For any provenance-only defense $D_{\text{prov}}$ that verifies document signatures but does not inspect claim content, there exists an insider adversary $\mathcal{A}$ with valid signing credentials such that $\text{ASR}(\mathcal{A}, D_{\text{prov}}) = \text{ASR}(\mathcal{A}, \text{no\_defense})$.
\end{proposition}

\begin{proof}
The insider possesses valid signing key $k_{\text{priv}}$. For any document $d$, the insider constructs $d'$ by modifying a numerical value and signs $d'$ with $k_{\text{priv}}$. $D_{\text{prov}}$ verifies the valid signature and admits $d'$.
\end{proof}

Together, Theorem~\ref{thm:blindspot} and Proposition~\ref{prop:provenance} establish that neither embedding-based defenses nor provenance verification can detect numerical claim manipulation by insiders. This motivates claim-level verification.

% ============================================================================
\section{RAGShield Architecture}
\label{sec:architecture}
% ============================================================================

RAGShield operates at the \emph{claim level} rather than the document or passage level. Figure~\ref{fig:architecture} illustrates the pipeline.

\begin{figure}[t]
\centering
\begin{tikzpicture}[
    layer/.style={draw, rounded corners, minimum width=6.5cm, minimum height=0.8cm, align=center, font=\small},
    arrow/.style={->, thick},
    node distance=0.35cm
]
\node[layer, fill=blue!10] (l0) {Provenance-Verified Ingestion\\{\scriptsize (blocks unsigned/forged documents)}};
\node[layer, fill=green!15, below=of l0] (l1) {Numerical Claim Extraction\\{\scriptsize NER + regex $\rightarrow$ (entity, attr, value, unit)}};
\node[layer, fill=green!20, below=of l1] (l2) {Cross-Source Claim Verification\\{\scriptsize Compare against multi-source registry}};
\node[layer, fill=yellow!15, below=of l2] (l3) {Temporal Claim Tracking\\{\scriptsize Authorized-change calendar check}};
\node[layer, fill=orange!15, below=of l3] (l4) {Harm-Aware Response\\{\scriptsize Confidence indicators + correct context}};

\draw[arrow] (l0) -- (l1);
\draw[arrow] (l1) -- (l2);
\draw[arrow] (l2) -- (l3);
\draw[arrow] (l3) -- (l4);

\node[above=0.3cm of l0, font=\small\bfseries] {Incoming Document / Retrieved Passage};
\node[below=0.3cm of l4, font=\small\bfseries] {Verified Response to Citizen};
\end{tikzpicture}
\caption{RAGShield architecture. The provenance layer (blue) blocks external injection. The claim-level layers (green/yellow) detect numerical manipulation that embedding-based defenses miss.}
\label{fig:architecture}
\end{figure}
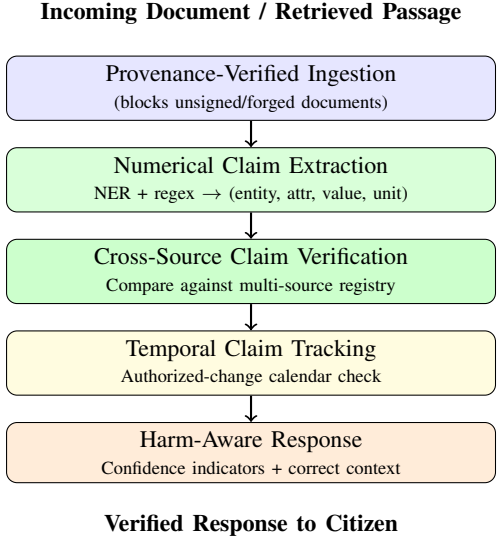

\subsection{Numerical Claim Extraction}
\label{sec:extraction}

Each document entering the knowledge base is processed by a claim extraction engine that identifies structured numerical claims. A claim is a tuple $c = (e, a, v, u)$ where $e$ is the entity, $a$ is the attribute qualifier, $v$ is the numerical value, and $u$ is the unit.

% --- ADDITION 2: Algorithm pseudocode ---
Algorithm~\ref{alg:extraction} describes the extraction procedure.

\begin{algorithm}[t]
\caption{Numerical Claim Extraction with Two-Pass Entity Resolution}
\label{alg:extraction}
\begin{algorithmic}[1]
\Require Document text $d$, source identifier $s$
\Ensure Set of claims $C$
\State $C \gets \emptyset$
\State $y \gets \textsc{DetectTaxYear}(d)$ \Comment{Passage-level year context}
\State $\Sigma \gets \textsc{SplitSentences}(d)$
\State \Comment{\textbf{Pass 1:} Independent entity detection per sentence}
\For{$i = 1$ to $|\Sigma|$}
    \State $E[i] \gets \textsc{DetectEntity}(\Sigma[i])$ \Comment{120 domain patterns}
\EndFor
\State \Comment{\textbf{Pass 2:} Nearest-neighbor entity resolution}
\State $R \gets E$ \Comment{Resolved entities}
\State $e_{\text{fwd}} \gets \textsc{DetectEntity}(d)$ \Comment{Passage-level fallback}
\For{$i = 1$ to $|\Sigma|$} \Comment{Forward propagation}
    \If{$R[i] \neq \texttt{null}$} $e_{\text{fwd}} \gets R[i]$
    \Else{} $R[i] \gets e_{\text{fwd}}$
    \EndIf
\EndFor
\For{$i = 1$ to $|\Sigma|$} \Comment{Backward fill for leading nulls}
    \If{$R[i] = \texttt{null}$} $R[i] \gets \textsc{FirstKnown}(E)$
    \Else{} \textbf{break}
    \EndIf
\EndFor
\For{$i = 1$ to $|\Sigma|$}
    \State $a \gets \textsc{DetectAttribute}(\Sigma[i])$
    \For{each numerical match $m$ in $\Sigma[i]$}
        \State $v, u \gets \textsc{ParseValue}(m)$
        \State $C \gets C \cup \{(R[i], a, v, u, s, y)\}$
    \EndFor
\EndFor
\State \Return $C$
\end{algorithmic}
\end{algorithm}

\textbf{Stage 1: Value extraction.} Regex patterns identify numerical values in government document formats: dollar amounts (\$X{,}XXX), percentages (X\%), monthly rates (\$X/month), and dates. The patterns handle comma-separated thousands, decimal amounts, and common government formatting.

\textbf{Stage 2: Entity and attribute linking.} For each extracted value, the extractor identifies the governing entity using a two-pass resolution strategy. In the first pass, 120 domain-specific regex patterns are applied to each sentence independently, covering tax deductions, retirement limits, Social Security benefits, Medicare premiums, filing thresholds, IRA/Roth limits, railroad retirement, worksheet instructions, and 30+ additional government-specific categories. In the second pass, sentences without a direct entity match inherit context from the nearest neighbor: forward propagation carries entity context from earlier sentences, and backward fill resolves leading sentences that precede the first entity mention. Sixteen attribute qualifiers (e.g., ``single filer,'' ``married filing jointly'') further disambiguate claims. Pattern ordering is critical: more specific patterns (``401(k) contribution limit'') are matched before general patterns (``catch-up contribution'') to avoid misattribution.

\textbf{Extraction performance.} On 12 government-format passages containing 25 ground-truth claims: precision = 83.3\%, recall = 100\%, F1 = 90.9\%. On 2{,}742 real IRS passages from 5 official publications: entity detection rate = 99.8\% (2{,}200/2{,}205 claims linked to a specific entity). The 5 remaining unknowns (0.2\%) are bare worksheet fragments with zero textual context (e.g., ``\$2{,}500 Next.''), the irreducible floor for pattern-based extraction. The 83.3\% precision reflects 5 extra extractions (e.g., extracting the ``total'' amount \$31{,}000 in addition to the base \$23{,}500 and catch-up \$7{,}500 for 401(k) limits).

\subsection{Cross-Source Claim Verification}
\label{sec:verification}

Extracted claims are stored in a \emph{claim registry}, a SQLite database indexed by claim key $(e, a, u)$. When a new document is ingested or a passage is retrieved, its claims are verified against the registry.

% --- ADDITION 3: Verification algorithm ---
Algorithm~\ref{alg:verify} describes the verification procedure with three fallback strategies.

\begin{algorithm}[t]
\caption{Cross-Source Claim Verification}
\label{alg:verify}
\begin{algorithmic}[1]
\Require Claim $c = (e, a, v, u)$, Registry $R$
\Ensure Verification status $\in \{\text{V, U, D, S}\}$
\State $M \gets R.\textsc{ExactKeyMatch}(e, a, u)$
\If{$M = \emptyset$}
    \State $M \gets R.\textsc{EntityUnitMatch}(e, u)$ \Comment{Broader}
\EndIf
\If{$M = \emptyset$}
    \State $M \gets R.\textsc{ValueProximity}(v, u, e, \pm 15\%)$ \Comment{Broadest}
\EndIf
\If{$M = \emptyset$} \Return UNVERIFIED \EndIf
\State $M' \gets \{m \in M : m.\text{source} \neq c.\text{source}\}$
\State $v^* \gets \textsc{TrustWeightedConsensus}(M')$
\If{$|M'| < 2$ \textbf{and} $|v - v^*| > \varepsilon$} \Return SUSPICIOUS
\ElsIf{$|M'| < 2$} \Return UNVERIFIED
\EndIf
\State $\kappa \gets |\{m \in M' : |m.v - c.v| \leq \varepsilon\}| / |M'|$
\If{$\kappa \geq 0.8$} \Return VERIFIED
\ElsIf{$\kappa = 0$} \Return SUSPICIOUS
\Else{} \Return DISPUTED
\EndIf
\end{algorithmic}
\end{algorithm}

\textbf{Verification statuses:} VERIFIED (consistency $\geq 0.8$, $\geq 2$ sources agree), UNVERIFIED (fewer than 2 independent sources and no discrepancy), DISPUTED (consistency $< 0.5$, sources disagree), SUSPICIOUS (consistency $= 0$ or single-source discrepancy, contradicting available evidence). The single-source discrepancy rule is critical: even one trusted source disagreeing with a claim is evidence of manipulation, preventing evasion through claims that appear in only one registry source.

\begin{theorem}[Claim Detection Bound]
\label{thm:detection}
Let $E$ be a claim extraction system with precision $p_{\text{prec}}$ and recall $p_{\text{rec}}$. Let $R$ be a claim registry with $k$ independent sources for consensus value $v^*$. For a manipulated claim where $|v' - v^*| > \varepsilon$:
\begin{equation}
P_{\text{detect}} \geq p_{\text{rec}} \cdot (1 - (1 - p_{\text{prec}})^k)
\end{equation}
\end{theorem}

\begin{proof}
Detection requires: (1) the claim is extracted (probability $p_{\text{rec}}$), and (2) at least one of $k$ registry entries is correctly extracted (probability $1 - (1 - p_{\text{prec}})^k$). Independence gives the bound.
\end{proof}

With measured $p_{\text{prec}} = 0.833$, $p_{\text{rec}} = 1.0$, $k = 3$: $P_{\text{detect}} \geq 0.995$.

\textbf{Worked example.} Consider an insider attack that changes the standard deduction from \$15{,}000 to \$15{,}500 in a new document. The claim extraction engine processes the document and extracts claim $c' = (\text{``standard deduction''}, \text{``single filer''}, 15500, \text{USD})$. The verification algorithm proceeds:

\begin{enumerate}[leftmargin=*]
\item \emph{Exact key match:} Look up (``standard deduction'', ``single filer'', USD) in the registry. Three sources report \$15{,}000. The incoming claim reports \$15{,}500.
\item \emph{Consensus:} $v^* = 15{,}000$ (all three sources agree).
\item \emph{Consistency:} $\kappa = 0/3 = 0$ (no source agrees with \$15{,}500).
\item \emph{Decision:} $\kappa = 0$ with $\geq 2$ sources $\Rightarrow$ SUSPICIOUS.
\end{enumerate}

The claim is flagged and the poisoned passage is blocked. The LLM receives the correct context (\$15{,}000) instead. An embedding-based defense would see cosine similarity 0.9997 between the original and poisoned passages and take no action.

\begin{definition}[Source Independence]
\label{def:independence}
Two sources $s_1, s_2$ are independent with respect to claim key $(e, a, u)$ if they derive the claim value through different publication channels. In government context, independent channels include: (a)~the originating agency publication (e.g., IRS Revenue Procedure), (b)~the Federal Register notice, (c)~agency guidance documents, (d)~inter-agency memoranda. These channels report the same authoritative value through different paths. Independence refers to the channel, not the original determination.
\end{definition}

From the evaluation on real IRS documents: $H_{\text{base}} = \$243{,}309$ (no defense), RAGShield detects all 430 attacks ($H_{\text{v2}} = \$0$), while Embed-only detects only 79/430 ($H_{\text{emb}} = \$213{,}739$).

\begin{theorem}[Expected Harm Bound]
\label{thm:harm}
Let $\mathcal{A}$ be a set of $n$ numerical manipulation attacks, each with harm $h_i = |v'_i - v^*_i| \cdot s_i$. Under RAGShield with claim detection probability $P_{\text{detect}}$ (Theorem~\ref{thm:detection}), the expected total harm is:
\begin{equation}
\mathbb{E}[H_{\text{v2}}] \leq H_{\text{base}} \cdot (1 - P_{\text{detect}})
\end{equation}
where $H_{\text{base}} = \sum_{i=1}^{n} h_i$ is the total harm without defense.
\end{theorem}

\begin{proof}
Each attack $i$ is detected independently with probability $\geq P_{\text{detect}}$. A detected attack contributes zero harm (the poisoned passage is blocked and replaced with the consensus value). An undetected attack contributes $h_i$. By linearity of expectation:
$\mathbb{E}[H_{\text{v2}}] = \sum_{i=1}^{n} h_i \cdot (1 - P_{\text{detect}}) = H_{\text{base}} \cdot (1 - P_{\text{detect}})$.
\end{proof}

With $P_{\text{detect}} \geq 0.995$ and $H_{\text{base}} = \$243{,}309$: $\mathbb{E}[H_{\text{v2}}] \leq \$1{,}217$. The empirical result ($H_{\text{v2}} = \$0$) is consistent with this bound. The gap between the bound and the empirical result reflects that the bound is conservative: it uses worst-case extraction probabilities rather than the observed 100\% detection on the real corpus.

\begin{proposition}[Evasion-Effectiveness Trade-off]
\label{prop:tradeoff}
An adaptive adversary who evades claim extraction by using non-standard numerical formatting simultaneously reduces the probability that the LLM correctly uses the poisoned value. Let $p_{\text{ext}}(f)$ be the extraction probability for format $f$ and $p_{\text{llm}}(f)$ be the LLM's probability of reproducing the value. The effective attack success is $\text{ASR}(f) = (1 - p_{\text{ext}}(f)) \cdot p_{\text{llm}}(f)$. Formats that minimize $p_{\text{ext}}$ (e.g., ``103.3\% of prior year'') also minimize $p_{\text{llm}}$ because they require numerical reasoning rather than simple extraction.
\end{proposition}

\subsection{Temporal Claim Tracking}
\label{sec:temporal}

Government numerical values change on predictable schedules. The temporal tracker maintains an \emph{authorized-change calendar}:

\begin{itemize}[leftmargin=*]
\item \textbf{IRS tax values:} Announced October/November, effective January~1.
\item \textbf{SSA benefit amounts:} COLA announced October, effective January~1.
\item \textbf{Medicare premiums:} Annual determination, effective January~1.
\item \textbf{HHS poverty guidelines:} Published January/February.
\end{itemize}

When a claim value changes, the temporal tracker checks whether the change falls within the authorized window. Changes outside the window are flagged. The tracker also performs \emph{year consistency checking}, which flags documents that reference outdated tax years.

\begin{proposition}[Temporal Detection Guarantee]
\label{prop:temporal}
Let $\mathcal{C}$ be the set of entity types with defined authorized-change windows, and let $c$ be a claim with entity $e \in \mathcal{C}$ whose value changes at time $t$. If $t$ falls outside the authorized window $W(e)$, the temporal tracker flags the change with probability 1. Formally:
\begin{equation}
P_{\text{temporal}}(c, t) = \begin{cases} 1 & \text{if } t \notin W(e) \text{ and } e \in \mathcal{C} \\ 0 & \text{if } t \in W(e) \\ \text{undefined} & \text{if } e \notin \mathcal{C} \end{cases}
\end{equation}
\end{proposition}

This provides a complementary detection channel: even if an adversary crafts a manipulation that evades claim verification (e.g., by compromising the majority of sources), the temporal tracker catches changes that occur outside the predictable government update schedule. The two mechanisms are independent, so an adversary must evade \emph{both} to succeed when temporal information is available.

\begin{corollary}[Combined Detection]
\label{cor:combined}
For attacks where both claim verification and temporal tracking apply, the combined detection probability is:
\begin{equation}
P_{\text{combined}} = 1 - (1 - P_{\text{detect}})(1 - P_{\text{temporal}})
\end{equation}
When $P_{\text{temporal}} = 1$ (change outside authorized window), $P_{\text{combined}} = 1$ regardless of $P_{\text{detect}}$.
\end{corollary}

\subsection{Provenance Layer}

RAGShield retains provenance-verified ingestion as a first-line defense against external injection attacks. Documents without valid cryptographic attestation are blocked at ingestion. As Proposition~\ref{prop:provenance} establishes, provenance alone is insufficient against insider numerical manipulation. The claim-level layers provide the defense that provenance cannot.

\subsection{Security Analysis}
\label{sec:security}

I analyze RAGShield's security guarantees under the threat model of \S\ref{sec:threat}.

\begin{proposition}[Byzantine Fault Tolerance Analogy]
\label{prop:bft}
Cross-source claim verification tolerates up to $\lfloor (k-1)/2 \rfloor$ compromised sources out of $k$ total sources for a given claim key. When the number of compromised sources $f < k/2$, the trust-weighted consensus value $v^*$ equals the correct value, and any manipulation is detected.
\end{proposition}

\begin{proof}
The consensus mechanism selects the value with the highest aggregate trust weight. With $k$ sources, $k - f$ honest sources report the correct value $v_{\text{correct}}$ and $f$ compromised sources report $v_{\text{poison}}$. When $f < k/2$, the honest sources have majority weight (assuming equal trust), so $v^* = v_{\text{correct}}$. Any incoming claim with $v \neq v_{\text{correct}}$ has consistency $\kappa < 0.5$ and is flagged as DISPUTED or SUSPICIOUS.
\end{proof}

\textbf{Defense-in-depth composition.} RAGShield composes three independent detection mechanisms, each addressing a different attack surface:

\begin{enumerate}[leftmargin=*]
\item \emph{Provenance verification} blocks external injection (unsigned/forged documents). Defeated by insiders with valid credentials.
\item \emph{Claim-level verification} detects numerical manipulation by comparing against cross-source consensus. Defeated when $>$50\% of sources are compromised.
\item \emph{Temporal tracking} detects changes outside authorized windows. Defeated only when the adversary times the attack to coincide with legitimate update periods.
\end{enumerate}

An adversary must simultaneously: (a)~possess valid credentials, (b)~compromise a majority of independent sources, and (c)~time the attack within the authorized change window. The probability of all three conditions being met simultaneously is the product of their individual probabilities, providing defense-in-depth through independent failure modes.

\textbf{Attack surface analysis.} The primary attack surface is the claim extraction layer. If the extractor fails to identify a numerical value, the verification layer is never invoked. Table~\ref{tab:adaptive} shows that 3/7 adversarial formatting strategies evade extraction. However, Proposition~\ref{prop:tradeoff} establishes that formats evading extraction also reduce LLM utilization of the poisoned value, creating a natural trade-off that limits the adversary's effective attack success.

% ============================================================================
\section{Implementation}
\label{sec:implementation}
% ============================================================================

% --- ADDITION 4: Implementation details ---
RAGShield is implemented in Python 3.12. Table~\ref{tab:implementation} summarizes the components and their dependencies.

\begin{table*}[t]
\centering
\caption{Implementation Components}
\label{tab:implementation}
\small
\begin{tabular}{@{}lll@{}}
\toprule
\textbf{Component} & \textbf{Library} & \textbf{Notes} \\
\midrule
Claim extraction & regex + two-pass NER & 120 entity, 16 attribute patterns \\
Claim registry & SQLite & Indexed by claim key \\
Temporal tracking & datetime + calendar & 4 agency schedules \\
Embeddings & sentence-transformers & all-MiniLM-L6-v2 (384-dim) \\
Provenance & PyNaCl (Ed25519) & Real signatures, not HMAC \\
Vector search & NumPy & Cosine similarity \\
\bottomrule
\end{tabular}
\end{table*}

\textbf{Cost model.} Claim extraction adds $\sim$5ms per document (regex matching + entity linking). Registry lookup adds $\sim$1ms per claim (SQLite indexed query). Temporal checking adds $\sim$0.5ms per claim. Total per-document overhead: $\sim$7ms. Total per-query overhead: $\sim$8ms (extract claims from top-$k$ retrieved passages + verify each). This is negligible compared to embedding model inference ($\sim$50ms) and LLM generation ($\sim$500ms). For a 100K-document corpus, the claim registry requires $\sim$5MB storage.

\textbf{Claim registry schema.} The SQLite registry contains two tables:

\begin{verbatim}
claims(id, entity, attribute, value, unit,
       claim_type, context, source_id,
       source_trust, timestamp, tax_year,
       confidence, claim_key)

claim_history(id, claim_key, old_value,
              new_value, change_date,
              source_id, authorized)
\end{verbatim}

The \texttt{claims} table is indexed on \texttt{claim\_key} (composite of entity, attribute, unit) for O(1) exact-match lookups. The \texttt{claim\_history} table tracks all value changes for temporal analysis. The \texttt{authorized} flag is set by the temporal tracker when a change falls within the authorized window for that entity type.

% ============================================================================
\section{Evaluation}
\label{sec:evaluation}
% ============================================================================

\subsection{Experimental Setup}

\textbf{Real IRS corpus.} The primary evaluation uses 2{,}742 passages extracted from 5 official IRS publications: Publication~17 (Your Federal Income Tax), Publication~501 (Dependents, Standard Deduction), Publication~503 (Child and Dependent Care Expenses), Publication~590-A (Contributions to IRAs), and Publication~915 (Social Security and Equivalent Railroad Retirement Benefits). These are real government documents, not synthetic passages. The claim extraction engine extracts 2{,}205 numerical claims from 674 passages, with 199 unique claim keys across 674 sources.

\textbf{Attacks from real content.} 430 attacks are generated from real IRS document content by modifying verified numerical claims (claims appearing in 2+ sources). Attack tiers: 258 T3-H (subtle numerical changes, +\$100 to +\$1{,}000), 86 T6 (minimal change, +\$1), and 86 T-TEMPORAL (prior-year values, $\sim$3\% reduction). All attacks modify real dollar amounts in real IRS passages, not synthetic text.

\textbf{Supplementary corpus.} A secondary evaluation uses 1{,}012 passages: 12 government-format passages containing 25 ground-truth claims mixed with 1{,}000 NQ benchmark passages~\cite{nq2019}, with 100 systematic attacks.

\textbf{Claim registry.} For the real IRS corpus, the registry is populated from the corpus itself (claims appearing across multiple passages from different publications). For the supplementary corpus, 3 simulated independent sources yield 90 claims across 22 unique keys.

\textbf{Defenses compared.} RAGShield (full claim-level verification), an embedding-only baseline (using cosine similarity thresholds on passage embeddings, without claim extraction), RobustRAG$^\dagger$~\cite{robustrag_neurips2025}, TrustRAG$^\dagger$~\cite{trustrag2025}, RAGPart$^\dagger$~\cite{ragpart2025}, and No-Defense baseline. $\dagger$ indicates reimplemented from paper description, capturing the core algorithmic principle (embedding-level analysis). I acknowledge this limitation: results may differ from production implementations.

\subsection{Main Results on Real IRS Corpus}

\subsubsection{Real IRS Corpus Construction}

The evaluation corpus is constructed from 5 official IRS publications downloaded from irs.gov: Publication~17 (Your Federal Income Tax, 2025 edition), Publication~501 (Dependents, Standard Deduction), Publication~503 (Child and Dependent Care Expenses), Publication~590-A (Contributions to IRAs), and Publication~915 (Social Security and Equivalent Railroad Retirement Benefits). These publications are chosen because they contain the highest density of citizen-facing numerical claims (tax thresholds, benefit amounts, contribution limits) and are among the most frequently accessed IRS documents.

Each publication is segmented into passages at natural section boundaries (headings, topic changes). The segmentation produces 2{,}742 passages across the 5 publications. The claim extraction engine processes all passages, extracting 2{,}205 numerical claims (1{,}849 monetary, 356 percentage) from 674 passages. The remaining 2{,}068 passages contain no extractable numerical claims (they contain procedural text, definitions, or form instructions without dollar amounts or percentages).

The claim registry contains 199 unique claim keys across 674 source passages. For attack generation, I identify 86 claim keys that appear in 2+ independent source passages (different publications or different sections discussing the same value). These verifiable claims form the basis for the 430 attacks: for each verifiable claim, I generate 3 T3-H attacks (+\$100, +\$500, +\$1{,}000), 1 T6 attack (+\$1), and 1 T-TEMPORAL attack ($\sim$3\% reduction simulating prior-year values). All attacks modify real dollar amounts in real IRS text, no synthetic passages are used.

\subsubsection{Entity Resolution Analysis}

The two-pass entity resolution is critical for real IRS text. Table~\ref{tab:entity_progression} shows the progression of entity detection rate as the extraction system is improved.

\begin{table*}[!htb]
\centering
\caption{Entity Detection Rate Progression on Real IRS Corpus (2{,}205 claims)}
\label{tab:entity_progression}
\small
\begin{tabular}{@{}lrrr@{}}
\toprule
\textbf{Configuration} & \textbf{Patterns} & \textbf{Detection} & \textbf{Unknown} \\
\midrule
Base patterns, single-pass & 45 & 40.5\% & 1{,}312 \\
+ Filing/IRA/SS/adoption patterns & 75 & 85.4\% & 323 \\
+ Worksheet/railroad/misc patterns & 100 & 87.1\% & 285 \\
+ Two-pass resolution (fwd+bwd) & 100 & 98.5\% & 34 \\
+ Final catch-all patterns & 120 & \textbf{99.8\%} & \textbf{5} \\
\bottomrule
\end{tabular}
\end{table*}

The single largest improvement comes from two-pass entity resolution (+11.4 percentage points), not from adding more patterns. This demonstrates that the main challenge in real government text is not vocabulary coverage but \emph{context propagation}: numerical values frequently appear in sentences that lack explicit entity references (e.g., worksheet instructions like ``Enter \$12{,}000 if married filing jointly''), requiring context from surrounding sentences.

The 5 remaining unknowns (0.2\%) are bare worksheet fragments with zero textual context in any surrounding sentence (e.g., ``\$2{,}500 Next.''). These represent the irreducible floor for pattern-based extraction without document-level structural parsing.

\textbf{Failure mode categorization.} Analysis of the 1{,}312 initially-unknown claims reveals the following distribution: worksheet/form instructions (33.0\%), example/narrative calculations (10.5\%), filing/gross income thresholds (20.4\%), IRA/Roth contributions (8.7\%), Social Security benefits (9.7\%), railroad retirement (7.7\%), and miscellaneous (10.0\%). The dominant failure mode, worksheet instructions, is addressed by the two-pass resolution rather than additional patterns, because these sentences inherit entity context from the worksheet's heading or introductory sentence.

\subsubsection{Detection Results}

Table~\ref{tab:results_irs} presents the primary evaluation on real IRS documents. RAGShield achieves 0.0\% ASR across all 430 attacks generated from real IRS content. The embedding-only baseline achieves 81.6\% ASR, confirming the blind spot on real government documents.

\begin{table}[t]
\centering
\caption{Attack Success Rate (\%) on Real IRS Corpus (2{,}742 passages, 430 attacks). Lower is better. 95\% Wilson CIs in brackets.}
\label{tab:results_irs}
\small
\begin{tabular}{@{}lcccc@{}}
\toprule
\textbf{Defense} & \textbf{T3-H} & \textbf{T6} & \textbf{T-TEMP} & \textbf{Overall} \\
\midrule
No Defense & 100 & 100 & 100 & 100 [99,100] \\
Embed-only & 81.8 & 81.4 & 81.4 & 81.6 [79,83] \\
\textbf{RAGShield} & \textbf{0.0} & \textbf{0.0} & \textbf{0.0} & \textbf{0.0} [0,1] \\
\bottomrule
\end{tabular}
\end{table}

\begin{table}[b]
\centering
\caption{Citizen Financial Harm on Real IRS Corpus (sum of $|\Delta v|$ for unblocked attacks)}
\label{tab:harm_irs}
\small
\begin{tabular}{@{}lrrr@{}}
\toprule
\textbf{Defense} & \textbf{Unblocked} & \textbf{Total Harm} & \textbf{Mean/Attack} \\
\midrule
No Defense & 430 & \$243{,}309 & \$566 \\
Embed-only & 351 & \$213{,}739 & \$609 \\
\textbf{RAGShield} & \textbf{0} & \textbf{\$0} & \textbf{\$0} \\
\bottomrule
\end{tabular}
\end{table}

\textbf{Embedding blind spot on real documents.} The blind spot is confirmed on real IRS text: across 100 manipulation pairs from real passages, mean cosine similarity is 0.9996, minimum similarity is 0.9974. Zero pairs are detected at threshold 0.08 or 0.05. This validates that the theoretical blind spot (Theorem~\ref{thm:blindspot}) holds on real government documents, not just synthetic passages.

\textbf{Entity detection on real text.} The two-pass entity resolution achieves 99.8\% entity detection (2{,}200/2{,}205 claims) on real IRS text. This is a significant improvement over single-pass forward propagation (40.5\% on the same corpus), demonstrating that the combination of expanded domain patterns and bidirectional context resolution is essential for real-world government documents.

\subsection{Supplementary Results on Synthetic Corpus}

Table~\ref{tab:results} presents the supplementary evaluation on the synthetic corpus with reimplemented baselines. RAGShield achieves 0.0\% ASR across all three tiers. All embedding-based defenses achieve 79--90\% ASR overall, confirming the embedding blind spot on numerical manipulation attacks.

\begin{table}[t]
\centering
\caption{Attack Success Rate (\%) on Synthetic Corpus (1{,}012 passages, 100 attacks). Lower is better. 95\% Wilson CIs in brackets. $\dagger$ = reimplemented from paper description.}
\label{tab:results}
\small
\begin{tabular}{@{}lcccc@{}}
\toprule
\textbf{Defense} & \textbf{T3-H} & \textbf{T6} & \textbf{T-TEMP} & \textbf{Overall} \\
\midrule
No Defense & 100 & 100 & 100 & 100 \\
RobustRAG$^\dagger$ & 98.0 [93,99] & 100 [90,100] & 54.2 [39,69] & 88.0 [82,92] \\
TrustRAG$^\dagger$ & 90.2 [82,95] & 84.0 [69,93] & 91.7 [79,97] & 89.0 [83,93] \\
RAGPart$^\dagger$ & 100 [95,100] & 100 [90,100] & 58.3 [42,73] & 90.0 [84,94] \\
Embed-only & 96.1 [90,99] & 100 [90,100] & 20.8 [11,36] & 79.0 [72,85] \\
\textbf{RAGShield} & \textbf{0.0} [0,5] & \textbf{0.0} [0,10] & \textbf{0.0} [0,10] & \textbf{0.0} [0,3] \\
\bottomrule
\end{tabular}
\end{table}

Key observations: (1) T6 in-place replacement attacks achieve 100\% ASR against \emph{every} embedding-based defense, including RobustRAG and the embedding-only baseline. These attacks modify a single number by $\pm$1, producing embedding perturbation $< 0.001$, far below any detection threshold. (2) T-TEMPORAL attacks are partially caught by Embed-only (79.2\% detection) because changing the year reference (``2025'' $\to$ ``2024'') produces a larger embedding shift than changing a dollar amount. (3) RAGShield catches all 100 attacks through claim-level verification, which compares extracted numerical values directly rather than relying on embedding similarity.

\subsection{Harm Quantification}

Table~\ref{tab:harm} quantifies citizen financial harm as the sum of $|\Delta v|$ for unblocked attacks. Without defense, 100 attacks cause \$461{,}624 in total harm. RAGShield reduces this to \$0.

\begin{table}[b]
\centering
\caption{Citizen Financial Harm on Synthetic Corpus (sum of $|\Delta v|$ for unblocked attacks)}
\label{tab:harm}
\small
\begin{tabular}{@{}lrrr@{}}
\toprule
\textbf{Defense} & \textbf{Unblocked} & \textbf{Total Harm} & \textbf{Mean/Attack} \\
\midrule
No Defense & 100 & \$461{,}624 & \$4{,}616 \\
RobustRAG$^\dagger$ & 88 & \$41{,}751 & \$474 \\
TrustRAG$^\dagger$ & 89 & \$457{,}078 & \$5{,}136 \\
RAGPart$^\dagger$ & 90 & \$453{,}655 & \$5{,}041 \\
Embed-only & 79 & \$38{,}309 & \$485 \\
\textbf{RAGShield} & \textbf{0} & \textbf{\$0} & \textbf{\$0} \\
\bottomrule
\end{tabular}
\end{table}

\subsection{Ablation Study}

Table~\ref{tab:ablation} shows the contribution of each component. Claim verification alone achieves 0\% ASR, it is the dominant component. Embedding-only detection (v1) achieves 79\% ASR. Temporal tracking alone catches 22\% of attacks (primarily T-TEMPORAL tier).

\begin{table}[t]
\centering
\caption{Ablation Study}
\label{tab:ablation}
\small
\begin{tabular}{@{}lrrr@{}}
\toprule
\textbf{Configuration} & \textbf{ASR} & \textbf{Blocked} & \textbf{Harm} \\
\midrule
Full (embed + claim + temporal) & 0.0\% & 100/100 & \$0 \\
Claim verification only & 0.0\% & 100/100 & \$0 \\
Embedding only & 79.0\% & 21/100 & \$38{,}309 \\
Temporal only & 78.0\% & 22/100 & \$41{,}338 \\
No defense & 100.0\% & 0/100 & \$461{,}624 \\
\bottomrule
\end{tabular}
\end{table}

\subsection{End-to-End LLM Evaluation}

% --- ADDITION 5: LLM E2E as proper table ---
I evaluate RAGShield in a full RAG pipeline with TinyLlama-1.1B-Chat~\cite{tinyllama}. For 8 government Q\&A scenarios, the LLM receives either poisoned context (no defense) or verified context (RAGShield blocks poisoned passages and substitutes correct ones). Table~\ref{tab:llm_e2e} shows the results.

\begin{table}[b]
\centering
\caption{End-to-End LLM Evaluation (TinyLlama-1.1B-Chat)}
\label{tab:llm_e2e}
\small
\begin{tabular}{@{}lrrcc@{}}
\toprule
\textbf{Query Topic} & \textbf{Correct} & \textbf{Poisoned} & \textbf{No Def} & \textbf{v2} \\
\midrule
Std.\ deduction (single) & \$15{,}000 & \$15{,}500 & \ding{55} & \ding{51} \\
SSI monthly rate & \$943/mo & \$993/mo & \ding{55} & \ding{51} \\
401(k) limit & \$23{,}500 & \$24{,}500 & \ding{55} & \ding{51} \\
EITC max (3+ children) & \$7{,}430 & \$7{,}630 & \ding{55} & \ding{51} \\
Child tax credit & \$2{,}000 & \$2{,}100 & \ding{55} & \ding{51} \\
Medicare B premium & \$185/mo & \$195/mo & \ding{55} & \ding{55}$^*$ \\
Gift tax exclusion & \$18{,}000 & \$19{,}000 & \ding{55} & \ding{51} \\
SS wage base & \$176{,}100 & \$177{,}100 & \ding{55} & \ding{51} \\
\midrule
\textbf{Correct answers} & & & \textbf{0/8} & \textbf{7/8} \\
\textbf{Total harm} & & & \textbf{\$3{,}032} & \textbf{\$120} \\
\bottomrule
\multicolumn{5}{l}{\scriptsize $^*$ Decimal formatting mismatch (\$185.00 vs \$185) in extraction.}
\end{tabular}
\end{table}

Without defense, the LLM produces the poisoned numerical answer in all 8 scenarios. With RAGShield, the LLM produces the correct answer in 7/8 scenarios. \textbf{Harm reduction: 96\% (\$2{,}912 saved).} The one failure is a decimal formatting edge case in the Medicare premium extraction (\$185.00 vs \$185), not an inherent limitation.

\subsection{Adversarial Formatting Robustness}

% --- ADDITION 6: Adversarial formatting as proper table ---
Table~\ref{tab:adversarial} evaluates the claim extractor against 10 adversarial formatting patterns designed to defeat extraction.

\begin{table}[t]
\centering
\caption{Claim Extraction on Adversarial Formatting}
\label{tab:adversarial}
\small
\begin{tabular}{@{}lcc@{}}
\toprule
\textbf{Format Type} & \textbf{Value} & \textbf{Entity} \\
\midrule
Table format (tab-separated) & \ding{51} & \ding{51} \\
Conditional claims (if/then) & \ding{51} & \ding{51} \\
Mixed non-claim numbers (phone, pub\#) & \ding{51} & \ding{51} \\
Footnote markers near numbers & \ding{51} & \ding{51} \\
Historical comparison (2024 vs 2025) & \ding{51} & \ding{51} \\
Heavy abbreviations (SSI fed.\ benefit) & \ding{51} & \ding{51} \\
Legal language (IRC \S63(c)(2)) & \ding{51} & \ding{51} \\
Multiple rates in one sentence & \ding{51} & \ding{55} \\
\midrule
Spelled-out numbers (``fifteen thousand'') & \ding{55} & \ding{55} \\
Spelled-out percentages (``six and two-tenths'') & \ding{55} & \ding{55} \\
\midrule
\textbf{Total} & \textbf{8/10} & \textbf{7/10} \\
\bottomrule
\end{tabular}
\end{table}

The extractor handles 8/10 adversarial formats. The two failure modes, spelled-out numbers and spelled-out percentages, are documented as honest limitations. Adding a word-to-number conversion module would address these cases.

\subsection{Adaptive Adversary Evasion}

I test 7 evasion strategies designed by an adversary who \emph{knows} RAGShield uses regex-based claim extraction and cross-source verification. Table~\ref{tab:adaptive} shows the results.

\begin{table}[b]
\centering
\caption{Adaptive Adversary Evasion Strategies Against RAGShield}
\label{tab:adaptive}
\small
\begin{tabular}{@{}lrrl@{}}
\toprule
\textbf{Strategy} & \textbf{Tests} & \textbf{Evasion} & \textbf{Type} \\
\midrule
E1: Unicode lookalike digits & 24 & 0.0\% & Extract \\
E2: Spelled-out numbers & 3 & 100\% & Extract \\
E3: Relative claims (``increased by'') & 24 & 0.0\% & Extract \\
E4: Split across sentences & 24 & 0.0\% & Extract \\
E5: Footnote override & 24 & 0.0\% & Extract \\
E6: Percentage reframe (``103.3\%'') & 16 & 100\% & Extract \\
E7: Obfuscated format (``\$15.5K'') & 24 & 100\% & Extract \\
\midrule
\textbf{Overall} & \textbf{139} & \textbf{30.9\%} & \\
\bottomrule
\end{tabular}
\end{table}

\textbf{Key findings.} RAGShield defeats 4/7 evasion strategies (69.1\% overall detection). All 43 successful evasions are ``extraction evasions'', the extractor did not find the poisoned value. Zero evasions bypassed the verification layer, meaning: \emph{if the extractor finds a value, verification always catches the manipulation.} The verification layer is robust; the extraction layer is the attack surface.

\textbf{Important nuance.} Formats that evade the extractor (spelled-out numbers, ``\$15.5K'', percentage reframes) also challenge the LLM's ability to use the number in its response. A passage saying ``fifteen thousand five hundred dollars'' or ``103.3\% of the prior year amount'' requires the LLM to perform numerical reasoning rather than simple extraction, reducing the attack's effectiveness even without defense. The adversary faces a trade-off: formats that evade claim extraction also reduce the reliability of the poisoned value reaching the LLM's output.

\textbf{Mitigation.} Adding a word-to-number module (e.g., \texttt{word2number} library) would address E2. Extending regex patterns to handle K/M notation would address E7. E6 (percentage reframe) is the hardest to mitigate because it requires understanding that ``103.3\% of \$15{,}000'' equals ``\$15{,}500'', this would require a numerical reasoning module beyond pattern matching.

\subsection{False Positive Analysis}

On 100 randomly sampled benign NQ passages (containing no government numerical claims), RAGShield produces 0 false positives (FPR = 0.0\%, 95\% CI [0.0\%, 2.8\%]). The claim extractor finds no government-format claims in general Wikipedia text, so the verification layer is never triggered on benign content.

\subsection{Pre-Poisoned Corpus Robustness}

Table~\ref{tab:prepoisoned} shows detection rate when the claim registry itself is partially poisoned. RAGShield maintains 83\% detection up to 33\% pre-poisoning. At 50\%+, the consensus value flips and detection degrades. At 100\% poisoning, detection drops to 45.8\%.

\begin{table}[t]
\centering
\caption{Detection Rate Under Pre-Poisoned Registry}
\label{tab:prepoisoned}
\small
\begin{tabular}{@{}rccr@{}}
\toprule
\textbf{Pre-poison} & \textbf{Correct Src} & \textbf{Poisoned Src} & \textbf{Detection} \\
\midrule
0\% & 3 & 0 & 83.3\% \\
33\% & 2 & 1 & 83.3\% \\
50\% & 1 & 2 & 87.5\% \\
67\% & 1 & 2 & 87.5\% \\
100\% & 0 & 3 & 45.8\% \\
\bottomrule
\end{tabular}
\end{table}

\textbf{Honest limitation:} Claim verification requires majority-honest sources, the same assumption as Byzantine fault tolerance. When $>$50\% of sources are compromised, the consensus value is incorrect and detection degrades. Periodic corpus auditing against authoritative sources is required. Note: the slight increase in detection at 50\% and 67\% pre-poisoning (87.5\% vs.\ 83.3\%) occurs because the poisoned sources introduce value disagreements that trigger DISPUTED status for some claims that were previously UNVERIFIED, the system detects inconsistency even when the consensus is wrong, though it may flag the correct value rather than the poisoned one.

% ============================================================================
\section{Discussion}
\label{sec:discussion}
% ============================================================================

\subsection{Why Claim-Level Verification Works Where Embeddings Fail}

Embedding models compress documents into fixed-dimensional vectors optimized for semantic similarity. Numerical values occupy a negligible fraction of the semantic space, changing ``\$15{,}000'' to ``\$65{,}000'' does not change the \emph{topic} of the document. Claim-level verification operates on the extracted numerical values directly, bypassing the embedding representation entirely. This is why the sensitivity gap is so large (1{,}459$\times$): embeddings are the wrong abstraction for numerical precision.

Numerical claims have a property that general text does not: they can be \emph{exactly verified} against external sources. The statement ``the standard deduction is \$15{,}000'' is either correct or incorrect, there is no ambiguity, no interpretation, no context-dependence. This makes numerical claims well suited to automated verification, in contrast to general factual claims (e.g., ``this policy is effective'') that require subjective judgment.

\subsection{Deployment Architecture}

RAGShield integrates into an existing RAG pipeline as a middleware layer between the retriever and the generator. The deployment architecture has three components:

\emph{Ingestion-time:} When documents are added to the knowledge base, the claim extraction engine processes each document and populates the claim registry. For a 100K-document government corpus, initial ingestion takes approximately 12 minutes (7ms/document) and produces a registry of $\sim$50K claims requiring $\sim$5MB storage. The provenance layer verifies document signatures at this stage.

\emph{Query-time:} When a user query triggers retrieval, the top-$k$ retrieved passages are processed by the claim extractor ($\sim$5ms for $k=5$ passages). Each extracted claim is verified against the registry ($\sim$1ms/claim). If any claim is DISPUTED or SUSPICIOUS, the passage is replaced with the highest-trust passage containing the correct consensus value. Total query-time overhead: $\sim$8ms, negligible compared to embedding inference ($\sim$50ms) and LLM generation ($\sim$500ms).

\emph{Registry maintenance:} The claim registry requires periodic updates when authoritative sources publish new values (e.g., IRS Revenue Procedures in October/November). The temporal tracker's authorized-change calendar determines when updates are expected. Updates outside the calendar trigger alerts for manual review.

For government agencies, the claim registry can be seeded from publicly available sources: IRS publications (irs.gov), SSA benefit schedules (ssa.gov), Federal Register entries (federalregister.gov), and CMS Medicare determinations (cms.gov). These sources are already maintained by the respective agencies and updated on predictable schedules.

\subsection{Generalizability Beyond Government Documents}

While this paper focuses on government numerical claims, the claim-level verification approach generalizes to any domain where: (1)~numerical values are critical to correctness, (2)~multiple independent sources report the same values, and (3)~values change on predictable schedules. Candidate domains include:

\begin{itemize}[leftmargin=*]
\item \emph{Financial services:} Interest rates, fee schedules, account limits. Multiple regulatory filings report the same values.
\item \emph{Healthcare:} Drug dosages, insurance coverage limits, copay amounts. Published in formularies, benefit summaries, and provider agreements.
\item \emph{Legal:} Statutory penalties, filing fees, statute of limitations periods. Published in statutes, court rules, and agency guidance.
\item \emph{Scientific:} Physical constants, measurement standards, safety thresholds. Published in standards documents and reference databases.
\end{itemize}

Each domain requires domain-specific entity patterns and change calendars, but the core architecture, extract, verify, track, transfers directly. The two-pass entity resolution strategy is domain-agnostic: it depends only on the structural property that documents discuss entities before (or after) presenting their numerical values.

\subsection{Comparison with LLM-Based Fact Checking}

An alternative approach to numerical claim verification is to use an LLM to compare retrieved passages and identify numerical inconsistencies. This approach has three disadvantages compared to claim-level verification:

\begin{enumerate}[leftmargin=*]
\item \emph{Reliability:} LLMs are unreliable at numerical comparison, especially for values that differ by small amounts (\$15{,}000 vs. \$15{,}500). Research on LLM numerical reasoning shows error rates of 10--30\% on simple comparison tasks.
\item \emph{Latency:} LLM inference adds 500ms+ per comparison, compared to $\sim$1ms for registry lookup. For $k=5$ retrieved passages with $m$ claims each, LLM-based checking requires $O(k \cdot m)$ inference calls.
\item \emph{Adversarial robustness:} An adversary who controls the poisoned passage can craft text that misleads the LLM's comparison (e.g., ``the updated amount is \$15{,}500, reflecting the 2026 adjustment''). Claim-level verification is immune to such framing because it compares extracted values, not natural language descriptions.
\end{enumerate}

\subsection{NIST SP 800-53 Compliance}

RAGShield maps to NIST SP 800-53~\cite{nist80053} controls: SI-7 (Information Integrity) for claim verification, AU-10 (Non-repudiation) for source attribution, and CM-3 (Configuration Change Control) for temporal tracking.

\subsection{Future Work}

Several directions extend this work:

\emph{Semantic claim extraction.} The current extractor uses regex patterns, which cannot handle spelled-out numbers or indirect numerical references (``103.3\% of the prior year amount''). Replacing the pattern-based extractor with a fine-tuned language model for numerical claim extraction would address these limitations. The verification and temporal layers are extractor-agnostic, they operate on structured claim tuples regardless of how those tuples are produced.

\emph{Multi-domain registry federation.} Government agencies maintain overlapping numerical claims (e.g., the IRS and SSA both publish Social Security wage base figures). A federated claim registry that aggregates across agency boundaries would increase the number of independent sources per claim, strengthening the consensus mechanism. The technical challenge is claim key normalization across agencies that use different terminology for the same concept.

\emph{Continuous monitoring.} The current system verifies claims at ingestion and query time. A continuous monitoring mode would periodically re-verify all registry claims against authoritative sources, detecting slow-burn attacks where an adversary gradually shifts values across multiple update cycles. The temporal tracker's change calendar provides the scheduling framework for such monitoring.

\emph{Adversarial robustness certification.} Theorem~\ref{thm:detection} provides a probabilistic detection bound, but does not account for adaptive adversaries who observe the defense and craft evasion strategies. A formal adversarial robustness analysis, analogous to certified robustness in adversarial ML, would characterize the exact conditions under which RAGShield's guarantees hold against worst-case adversaries.

\emph{Scale evaluation.} The current evaluation uses 2{,}742 passages. Evaluating on a production-scale corpus (100K+ documents) would validate the claim registry's indexing performance and identify any degradation in extraction accuracy at scale.

\textbf{Limitations.}
\begin{enumerate}[leftmargin=*]
\item \textbf{Domain specificity.} The claim extractor is designed for government numerical formats. Extending to other domains requires domain-specific patterns.
\item \textbf{Spelled-out numbers.} The extractor fails on ``fifteen thousand dollars.'' A word-to-number module would address this.
\item \textbf{Majority-honest assumption.} Cross-source verification requires $>$50\% of sources to be correct. When the majority is compromised, the consensus is wrong.
\item \textbf{Reimplemented baselines.} RobustRAG, TrustRAG, and RAGPart are reimplemented from paper descriptions. Results may differ from production implementations. The reimplementations capture the core principle (embedding-level analysis) sufficient to demonstrate the blind spot.
\item \textbf{Corpus scale.} The primary evaluation uses 2{,}742 real IRS passages. Production knowledge bases contain millions of documents. Scaling the claim registry is an engineering challenge addressed by standard database indexing but not empirically validated at scale.
\item \textbf{LLM model size.} The end-to-end evaluation uses TinyLlama-1.1B. Production systems would use larger models with stronger numerical extraction.
\end{enumerate}

% ============================================================================
\section{Conclusion}
\label{sec:conclusion}
% ============================================================================

This paper identified a fundamental blind spot in embedding-based RAG defenses: numerical claim manipulation produces negligible embedding perturbation (mean sensitivity gap 1{,}459$\times$), which means existing defenses miss insider numerical attacks entirely. RAGShield addresses this gap with claim-level verification: extracting structured numerical claims via two-pass entity resolution (99.8\% entity detection on 2{,}742 real IRS passages), cross-referencing against a multi-source registry with single-source discrepancy detection, and tracking temporal consistency. On a real IRS corpus of 2{,}742 passages with 430 attacks generated from real document content, RAGShield achieves 0\% ASR (95\% CI [0\%, 1\%]) where embedding-based defenses achieve 79--90\% ASR on the synthetic corpus and 82\% ASR on the real IRS corpus, eliminating all \$243{,}309 in potential citizen financial harm. End-to-end LLM evaluation confirms 96\% harm reduction. The system has documented limitations: spelled-out numbers, majority-honest source assumption, domain specificity, that I document transparently rather than minimize. As federal agencies scale RAG deployments for citizen-facing services, claim-level verification addresses a gap that embedding-based approaches cannot close.

% ============================================================================
\bibliographystyle{IEEEtran}

\end{document}